\documentclass{article}
\usepackage{ijcai19}

\usepackage{booktabs}
\usepackage{times}
\usepackage{helvet}
\usepackage{courier}
\usepackage{amsmath,amsthm}
\usepackage{amssymb,amsfonts,textcomp}
\usepackage{latexsym}
\usepackage[noend,ruled]{algorithm2e}
\usepackage{graphicx}
\usepackage{latexsym}
\usepackage{multirow}
\usepackage{rotating}
\usepackage{xcolor}
\usepackage{tikz}

\usepackage{amsthm}
\newtheorem{theorem}{Theorem}
\newtheorem{lemma}[theorem]{Lemma}
\newtheorem{proposition}[theorem]{Proposition}
\newtheorem{corollary}[theorem]{Corollary}
\newtheorem{definition}{Definition}

\newcommand*{\citet}[1]{\citeauthor{#1} \shortcite{#1}}

\newcommand{\calS}{{{\mathcal G}}}
\newcommand{\calV}{{{\mathcal V}}}
\newcommand{\calE}{{{\mathcal E}}}
\newcommand{\calA}{{{\mathcal A}}}
\newcommand{\rmD}{{\mathrm D}}
\newcommand{\rmN}{{\mathrm N}}
\newcommand{\rmR}{{\mathrm R}}

\begin{document}

\title{Single-Crossing Implementation}  

\author{Nathann Cohen$^1$\and Edith Elkind$^2$
\and Foram Lakhani$^2$
\affiliations
$^1$CNRS, France\\
$^2$University of Oxford, UK}  

\maketitle

\begin{abstract}  
An election over a finite set of candidates 
is called single-crossing if, as we sweep through the list of voters
from left to right, the relative order of every pair of candidates changes at most once. 
Such elections have many attractive properties:
e.g., their majority relation is transitive and they admit efficient algorithms for problems that
are NP-hard in general. If a given election is not single-crossing, 
it is important to understand what are the obstacles that prevent it from having this property.
In this paper, we propose a mapping between elections and graphs
that provides us with a convenient encoding of such obstacles. 
This mapping enables us to use the toolbox of graph theory 
in order to analyze the complexity of detecting nearly 
single-crossing elections, i.e., elections that can be made
single-crossing by a small number of modifications.
\end{abstract}

\section{Introduction}\label{sec:intro}
As the deadline for Brexit approaches, the British newspapers clarify their positions 
on what they consider to be the best course of action\footnote{\tt https://bit.ly/2txNgpn}. The three 
possibilities currently under consideration are 
(1) accepting Theresa May's deal with the EU (D), 
(2) leaving the EU with no deal (N), 
and (3) postponing Brexit or canceling it altogether (R).
Among the major newspapers, 
the most pro-Leave position is taken by The Daily Telegraph, which is strongly critical of May's deal, 
and ranks it below no deal; it is also strongly opposed to any delays, so its ranking
can be described as $\rmN>\rmD>\rmR$. On the opposite end of the spectrum is The Guardian, which
backs remaining in the EU, but views May's deal as (scarcely) more acceptable than no deal at all, 
so its ranking is $\rmR>\rmD>\rmN$. The Times and The Sun take a more moderate position: both back May's deal, 
but The Times views no deal as the most catastrophic option, while The Sun is firmly opposed to any 
delays to Brexit. Thus, if we order the newspapers according to their political stance,
from left to right, the rankings change from $\rmR>\rmD>\rmN$ to $\rmD>\rmR>\rmN$ to 
$\rmD>\rmN>\rmR$ to $\rmN>\rmD>\rmR$.

These preferences have the property that for any two of the available
alternatives $X$, $Y$ if the first newspaper in our order
ranks $X$ over $Y$ then all newspapers that rank $X$ over $Y$ appear before all newspapers 
that rank $Y$ over $X$, 
i.e., every pair of alternatives `crosses' at most once.
An ordered list of rankings (an {\em election})
that has this property is known as {\em single-crossing}
(see Section~\ref{sec:prelim} for formal definitions). 
Single-crossing elections have many attractive properties
and have recently received a lot of attention in computational
social choice literature; see, e.g., the survey by \citet{ELP-trends}. 
Besides the example in the previous paragraph, 
there are real-life settings where we expect 
to observe essentially single-crossing preferences. For instance, when a country
is about to introduce a flat tax rate and the voters are asked to rank 
several options (say, 25\%, 28\%, 33\%, 45\%), they have to consider the tradeoff
between the amount they will have to pay and the value of the government services 
that can be provided at a given level of taxation; if all voters apply this reasoning, 
the election where the voters are ordered by income should be single-crossing
(this example dates back to \citet{mir:j:single-crossing}.

However, the examples we described, while single-crossing in spirit, may not be
single-crossing in a formal sense, because the single-crossing property is very fragile. 
For instance, in the Brexit scenario, if we expand the list of newspapers to include 
a broader sample of publications, we may find left-leaning newspapers that support no deal.
In the tax scenario, 
not all voters may be capable of evaluating the consequences of each choice.

Now, we can easily check whether an election where voters are ordered
according to a publicly observable parameter is single-crossing,
simply by looking at all pairs of candidates. However, 
as argued above, we expect the answer to be `no'.
A more important---and more challenging!---task is to understand whether this election is close to single-crossing, i.e., 
can be made single-crossing by deleting a few voters or candidates or swapping a few
pairs of candidates, as this will tell us whether the observable parameter
used to order the voters is relevant for understanding the voters' preferences. This knowledge is important,
e.g., for managing electoral campaigns, as it can be used to identify the segments
of voting population that are more likely to support a given candidate (which can 
help the campaign managers to decide whom to target in a get-out-the-vote effort). 

\smallskip

\noindent{\bf Our Contribution\ }\\
We introduce a mapping between elections and graphs that
enables us to use powerful graph-theoretic machinery to analyze 
nearly single-crossing elections. Briefly, given an election $E$
over a set of candidates $C$ (i.e., an ordered list of linear orders over $C$), 
we build an undirected graph $G$ that has $C$
as its set of vertices, and contains an edge between two candidates $a$ and $b$ 
if and only if $a$ and $b$ cross more than once in $E$; 
we say that $E$ {\em implements} $G$. In other words, 
the graph $G$ documents obstacles that prevent $E$ from being single-crossing;
in particular, an independent set in $G$ corresponds to a subset
of candidates such that the restriction of $E$ to this subset
is single-crossing. We then ask which graphs are $n$-implementable, 
i.e., can be implemented by elections with $n$ voters.
We show (Section~\ref{sec:n}) that we can obtain
any undirected graph in this manner; in fact, the number of voters required is bounded by a linear
function in the size of the graph. However, for any constant $n$ there are graphs
that are not $n$-implementable. In Section~\ref{sec:3}
we focus on $3$-implementable graphs and obtain
a complete characterization of this class of graphs by relating it to the class
of {\em permutation graphs}. We also argue that all $3$-implementable
graphs are {\em comparability graphs};
importantly, every comparability graph is a perfect graph. 

Our results have implications for the problem of deciding whether
an election is nearly single-crossing with respect to the given order of voters
(Section~\ref{sec:apps}). In particular, we use our mapping to establish
the hardness of computing two measures of how close a given election is 
to being single-crossing: one of these measures is based on deleting 
as few candidates as possible, and the other is based on splitting the candidates 
into as few groups as possible. On the other hand, 
our results for $3$-implementable graphs enable us
to show that the problems we consider are in P for elections with $3$ voters.

\smallskip

\noindent{\bf Related Work\ }\\
The concept of single-crossing elections has been proposed 
in the social choice literature several decades
ago~\cite{mir:j:single-crossing,rob:j:tax}. Single-crossing elections
are appealing both from a purely social choice-theoretic perspective and from 
a computational perspective: for instance, their weak majority relation
is necessarily transitive \cite{mir:j:single-crossing} 
and they admit efficient algorithms for determining
a winning committee under a well-known committee selection rule whose output
is hard to compute for general preferences \cite{sko-yu-fal:j:mwsc}.
Several groups of authors have considered the problem of identifying
nearly single-crossing elections \cite{bre-che-woe:j:nice,cor-gal-spa:c:spsc-width,elk-lac:c:detect,jae-pet-elk:c:nearly-sc},
but in all these papers the authors assumed that there was no publicly observable
parameter that determined the ordering of the voters, i.e., they
considered the problem of reordering the voters so that the resulting election
can be made single-crossing by applying a small number of modifications;
in contrast, we assume that the order of voters is fixed.

Our analysis is similar in spirit to the research on implementation of directed graphs 
as majority graphs. In this line of work, the input is a directed graph with a vertex 
set $C$, and the goal is to construct an election $E$ over the set of candidates $C$
such that there is a directed edge from $a\in C$ to $b\in C$ in the input graph
if and only if a strict majority of voters in $E$ prefer $a$ to $b$. The classic McGarvey 
theorem~\cite{mcgarvey} establishes that every directed graph can be implemented
in this way using at most two voters per edge, and subsequent work has reduced this number
to $\Theta(\frac{|C|}{\log |C|})$~\cite{stearns,erdos}; Corollary~\ref{cor:2n+1implementation} 
in Section~\ref{sec:n} can be viewed as an analogue of McGarvey's theorem in our setting. 
Recently, \citet{brandt-few} investigated what directed graphs 
can be implemented by elections with two or three voters; this research is similar 
to our analysis in Section~\ref{sec:3}. 
 
\section{Preliminaries}\label{sec:prelim}
For each $s\in{\mathbb N}$, we denote the set $\{1, \ldots, s\}$ by $[s]$.

\smallskip

\noindent{\bf Elections\ }
An \emph{election} is a pair $(C,V)$, where 
$C$ is a finite set of {\em candidates} and $V = (v_1, \ldots, v_n)$ is 
a list of {\em votes}. Each vote $v_i$, $i\in [n]$, is a linear order over $C$.
We refer to elements of $[n]$ as {\em voters}; thus, $v_i$ is the vote of voter $i$.
We will sometimes use the term `profile' to refer to $V$, and we use the terms
`vote', `preference' and `ranking' interchangeably. We 
say that voter $i$ {\em prefers} $a$ to $b$ or {\em ranks} $a$ over $b$ (and write $a\succ_i b$) 
if $a$ precedes $b$ in the linear order $v_i$. 
A {\em restriction} of an election $E=(C, V)$ 
with $V=(v_1, \dots, v_n)$ to a subset of
candidates $X\subseteq C$ is an election $E|_X=(X, V^X)$, 
where $V^X=(v_1^X, \dots, v_n^X)$ and for every pair of candidates $a, b\in X$
and every $i\in [n]$ it holds that $a$ is ranked above $b$ in $v_i^X$
if and only if $a$ is ranked above $b$ in $v_i$.

{\em Single-crossing} (also known as {\em intermediate} or {\em order-restricted}) preferences
capture settings where the voters can be ordered along a single axis
according to their preferences.  
\begin{definition}\label{def:sc}
  An election $E = (C,V)$ with  
  $V = (v_1, \ldots, v_n)$ is 
  {\em single-crossing} 
  if for every pair of candidates $\{a, b\}$
  with $a \succ_1 b$ there is a $t \in [n]$
  such that $\{ i \in [n] \mid a \succ_i b \} = [t]$.
\end{definition}
%
We emphasize that we define single-crossing elections with respect to a fixed order
of the voters, i.e., we are interested in settings where voters are ordered according to a
publicly observable parameter.

\smallskip

\noindent{\bf Graphs\ }
An {\em undirected graph} is a pair $G=(\calV, \calE)$, where $\calV$ is a finite set, 
and $\calE$ is a collection of size-$2$ subsets of $\calV$. The elements of $\calV$ are called {\em vertices},
and the elements of $\calE$ are called {\em edges}. 
For readability, we will sometimes write $ab$ instead of $\{a, b\}$.
%
We assume that the reader is familiar with the definitions of a {\em path}, a {\em cycle}, 
a {\em tree}, and a {\em bipartite graph}.
A {\em hole} is a cycle $a_1, \dots, a_k$ with $k\ge 4$ such that $a_ia_j\in \calE$
if and only if $|i-j|=1$ or $\{i, j\}=\{1, k\}$. An {\em anti-hole} is a 
sequence of distinct vertices $a_1, \dots, a_k$ with $k\ge 4$ 
such that $a_ia_j\not\in\calE$ if and only if $|i-j|=1$ or $\{i, j\}=\{1, k\}$. 
%

A {\em directed graph} is a pair $G=(\calV, \calA)$, where $\calV$ is a finite set,
and $\calA$ is a collection of ordered pairs of elements of $\calV$. An element of $\calA$
is called an {\em arc}; an arc $(a, b)$ {\em points} from vertex $a$ to vertex $b$.
An undirected graph $(\calV, \calE)$ can be turned into a directed graph $(\calV, \calE)$
by choosing an {\em orientation} for each edge $\{a, b\}\in \calE$, i.e., transforming $\{a, b\}$
into $(a, b)$ or $(b, a)$.
A directed graph $G=(\calV, \calA)$ is said to be {\em transitive} if for 
every triple of vertices $u, v, w\in\calV$ such that
$(u,v)\in \calA$ and $(v, w)\in \calA$ it holds that $(u, w)\in \calA$.


\section{Single-Crossing Implementation}
A pair of candidates $\{a, b\}$ 
is a {\em multi-crossing pair} in  
an election $E=(C, V)$ with $V=(v_1, \dots, v_n)$
if $a\succ_i b$, $b\succ_j a$, and $a\succ_k b$
 for some $i, j, k\in [n]$ with $i<j<k$.

\begin{definition}
The {\em multi-crossing graph} of an election $E=(C, V)$
is an undirected graph $\gamma(E) =(\calV, \calE)$ such that $\calV=C$
and $ab\in \calE$ if and only if $\{a, b\}$ is a multi-crossing pair in $E$.
An election $E=(C, V)$ {\em implements} an undirected graph $G=(\calV, \calE)$
if $G=\gamma(E)$.
We say that a graph $G$ is {\em $n$-implementable} if there exists an $n$-voter election that implements it.
Since the set of candidates in an election $(C, V)$ that implements a graph $G = (\calV, \calE)$
is necessarily $\calV$, we often omit $C$ from the notation and speak of a profile $V$ 
that implements $G$. 
\end{definition}

By definition, the only graph that is $2$-implementable is the graph with no edges.
Thus, in the remainder of the paper we study graphs that are $n$-implementable for $n\ge 3$.


\section{$\boldsymbol{3}$-Implementability}\label{sec:3}
To build the reader's intuition, we first consider $3$-implementation.
We show how to implement several families
of graphs, such as paths, trees and even-length cycles. 
While for some of these families their $3$-implementability follows from the more general results
in Section~\ref{sec:3general}, the proofs below provide efficient algorithms 
for finding a $3$-voter profile that implements a given graph.
Also, we relate $3$-implementable graphs
to other well-known classes of graphs, such as permutation graphs and comparability graphs.
Finally, we prove that some graphs are not $3$-implementable.

\subsection{Examples}\label{sec:3example}
First, it is easy to see that we can $3$-implement empty graphs and cliques:
an empty graph is implemented by a profile where all three voters are identical, 
and a clique can be implemented by a profile where the first and the third voter rank 
the candidates in the same order, and the second voter ranks the candidates in the opposite order.
%
%
A somewhat more complex construction establishes that all paths and even-length cycles 
are $3$-implementable.
\begin{proposition}\label{prop:path}
There is a polynomial-time algorithm that given a graph
$P=(\calV, \calE)$ that is a path or an even-length cycle
constructs a $3$-voter profile that implements $P$. 
\end{proposition}
\begin{proof}
Suppose that $P$ is a path. For convenience, we assume that $\calV=[s]$ and 
$\calE= \{\{i, i+1\}\mid i\in [s-1]\}$.
To start, we construct a $3$-voter profile where all voters rank the candidates as $1\succ 2 \succ\dots\succ s$.
Then, we modify the preferences of the first and the third voter by swapping candidates $2i$ and $2i+1$
in her rankings, for $i=1, \dots, \lfloor\frac{s-1}{2}\rfloor$. Also, we modify the preferences
of the second voter by swapping candidates $2i$ and $2i-1$ in her ranking, for 
$i=1, \dots, \lfloor\frac{s}{2}\rfloor$. Table~\ref{tbl:path} (left) illustrates the resulting profiles for $s=6$.
To see why this profile implements a path, consider an even-numbered candidate $2i$, $2i<s$. By construction, 
all voters rank $2i$ above all candidates $j$ with $j<2i-1$ and below all candidates $k$ with $k>2i+1$. 
On the other hand, voters 1 and 3 rank $2i$ below $2i+1$ and above $2i-1$, whereas voter 2
ranks $2i$ above $2i-1$ and below $2i+1$. Thus, $\{2i, j\}$ is a multi-crossing pair if and only if $j\in\{2i-1, 2i+1\}$.
For odd-numbered candidates $2i+1$ with $2i+1<s$, as well as for candidates $1$ and $s$, the argument is similar.  
\begin{table}
\begin{center}
\begin{tabular}{|c|c|c|}
\hline
$v_1$ & $v_2$ & $v_3$\\
\hline
1 &2 &1\\
3 &1 &3\\
2 &4 &2\\
5 &3 &5\\
4 &6 &4\\
6 &5 &6\\
\hline
\end{tabular}\hspace{2cm}
\begin{tabular}{|c|c|c|}
\hline
$v_1$ & $v_2$ & $v_3$\\
\hline
      1 & 2 & 3 \\ 
      3 & 4 & 5 \\
      2 & 3 & 4 \\
      5 & 6 & 1 \\
      4 & 5 & 2 \\
      6 & 1 & 6 \\
\hline
\end{tabular}
\end{center}
\caption{$3$-implementations of path and cycle of length $6$.\label{tbl:path}}
\end{table}

A similar approach can be used if $P$ is an even-length cycle;
we omit the proof, but provide an example in Table~\ref{tbl:path}(right).
\end{proof}

The reader may wonder why we only consider cycles of even length in Proposition~\ref{prop:path}. 
Now, the cycle of length $3$ is $3$-implementable because it is a clique. 
However, for $k\ge 2$ the cycle of length $2k+1$ is not $3$-implementable; this follows from 
Theorem~\ref{thm:3comparability} in Section~\ref{sec:3general}.

On the other hand, we can extend the result of Proposition~\ref{prop:path} to arbitrary trees.

\begin{proposition}\label{prop:tree}
There is a polynomial-time algorithm that given a graph
$T=(\calV, \calE)$ that is a tree
constructs a $3$-voter profile that implements $T$. 
\end{proposition}
\begin{proof} 
Given a tree $T=(\calV, \calE)$, we pick an arbitrary vertex $r_0\in\calV$ to be its root. 
Our implementation is recursive. We observe that an isolated vertex is trivially implementable.
Then we consider a vertex $r$ of $T$ whose children are $r_1,\dots,r_k$, $k\ge 1$. We show that, 
if for each $i\in [k]$ we have a $3$-implementation of the tree rooted at $r_i$ in which the first voter
ranks $r_i$ first, then we can construct a $3$-implementation of the tree rooted at $r$
in which the first voter ranks $r$ first. Using this idea, we can construct an implementation
of $T$ starting from the leaves and ending at $r_0$.

Fix a vertex $r\in \calV$ that is not a leaf. 
Let $r_1,\dots,r_k$ be the children of $r$,
and for each $i\in [k]$ let $T_i$ be the subtree of $T$ rooted at $r_i$;
let $\calV_i$ be the set of vertices of $T_i$.
Suppose that for each $i\in [k]$ we have 
a $3$-implementation of $T_i$ in which $r_i$ is ranked first in the first vote.

We will first stack these three implementations on top of each other:
we construct a $3$-voter profile where each voter ranks all candidates in $\calV_i$
above all candidates in $\calV_j$ for all $1\le i< j\le k$, and for each $t\in\{1, 2, 3\}$
and each $i\in [k]$ the $t$-th voter ranks $a\in\calV_i$ above $b\in\calV_i$
if and only if the $t$-th voter in the given $3$-implementation of $T_i$ ranks $a$ above $b$.
Then we pull the candidates $r_1, \dots, r_k$ to the top of the first vote: we modify the preferences
of voter $1$ so that she ranks $r_i$ in position $i$ for $i\in [k]$ and the relative order
of the other candidates remains unchanged. Note that this step does not introduce any multi-crossing pairs.

In remains to insert $r$ into the voters' rankings. To this end, in the first vote we insert $r$ after the first 
$k$ candidates, in the second vote we place $r$ on top,  and in the third vote we place $r$ last.
Note that for each $i\in [k]$ the pair $\{r, r_i\}$ is multi-crossing, but for every $i\in [k]$ and 
every $a\in \calV_i\setminus\{r_i\}$ the pair $\{r, a\}$ is not multi-crossing, as both of the first two voters
rank $r$ above $a$. Thus, we have implemented the edges connecting $r$ to its children.
However, in the resulting profile voter $1$ does not rank $r$ first.
To remedy this, we first reverse the order of candidates in each vote and then reverse the order of votes;
neither of these operations changes the set of multi-crossing pairs, and in the resulting
profile $r$ is ranked first. Each of these steps can be implemented in polynomial time.
\end{proof} 


\subsection{General Constructions}\label{sec:3general}
We can relate $3$-imple\-men\-table graphs 
to two well-known classes of graphs: permutation graphs and comparability graphs.

\begin{definition}[Permutation graph]
  An undirected graph $G=(\calV,\calE)$ is a {\em permutation graph} if there exist permutations 
  $\pi_1,\pi_2$ of $\calV$ such that $ab \in \calE$ if and only if 
  $a$ appears before $b$ in exactly one of the permutations $\pi_1$ and $\pi_2$.
\end{definition}

It can be decided in polynomial time whether a given graph is a permutation graph;
moreover, if the answer is `yes', the respective permutations can be constructed 
in polynomial time as well \cite{golumbicbook,cleanup}.

It is immediate that every permutation graph can be implemented by a $3$-voter profile.
\begin{theorem}\label{thm:permutation}
  Every permutation graph is $3$-implementable, and a $3$-voter profile that implements it
  can be computed in polynomial time.
\end{theorem}
\begin{proof}
Let $G$ be a permutation graph, and let
$\pi_1$ and $\pi_2$ be two permutations that witness this. 
Then the profile $(\pi_1, \pi_2, \pi_1)$
implements $G$. 
\end{proof}

In fact, the proof of Theorem~\ref{thm:permutation} suggests a stronger claim:
an undirected graph is a permutation graph if and only if it can be implemented
by a $3$-voter profile where the first and the third voter have the same preferences.
Recall that our implementation of cliques has this property, but our implementation
of even-length cycles does not. There is a reason for this: it is not hard to show that cycles
of length at least $5$ are not permutation graphs. 
%
%
Thus, permutation graphs 
form a proper subclass of $3$-implementable graphs. The following proposition
further clarifies the relationship between $3$-implementable graphs and
permutation graphs.

\begin{proposition}\label{prop:permintersect}
A graph $G= (\calV,\calE)$ is $3$-implementable if and only if there exist two permutation graphs 
$(\calV,\calE^1)$ and $(\calV,\calE^2)$ such that 
$\calE=\calE^1 \cap \calE^2$.
\end{proposition}
\begin{proof}
Let $(\calV,\calE^1)$ and $(\calV,\calE^2)$ be two permutation graphs on the same set of vertices $\calV$. 
Let $\pi_1,\pi_2$ (respectively, $\pi_3, \pi_4$) be a pair of permutations witnessing that 
$(\calV,\calE^1)$ (respectively, $(\calV, \calE^2)$) is a permutation graph.
Note that if a pair of permutations $\pi, \widehat\pi$ witnesses that a given graph is a permutation
graph, then so does the pair of permutations $\sigma\circ\pi, \sigma\circ\widehat\pi$, 
for any given permutation $\sigma$. Applying this observation to $\pi_3, \pi_4$ with $\sigma=\pi_2\circ\pi_3^{-1}$,
we can assume that $\pi_3=\pi_2$. 
Then the profile $(v_1, v_2, v_4)$ where for each $i\in\{1, 2, 4\}$
voter $i$ ranks the candidates according to $\pi_i$ implements $(\calV, \calE^1\cap\calE^2)$:
a pair of candidates $\{a, b\}$ is multi-crossing in this profile if and only if both $\pi_1$
and $\pi_4$ disagree with $\pi_2=\pi_3$ on the order of $a$ and $b$.

Conversely, let $(v_1,v_2,v_3)$ be a $3$-implementation of a graph $G =(\calV,\calE)$. 
Consider the graphs $(\calV, \calE^1)$ and $(\calV, \calE^2)$ implemented, respectively, 
by $(v_1,v_2,v_1)$ and $(v_2,v_3,v_2)$. As argued in the proof of Theorem~\ref{thm:permutation},
both of these graphs are permutation graphs, and a pair of candidates $\{a, b\}$ is multi-crossing
in $(v_1, v_2, v_3)$ if and only if $ab\in\calE^1$ and $ab\in \calE^2$.
This completes the proof.
\end{proof}

Another relevant class of graphs is {\em comparability graphs}.

\begin{definition}[Comparability graph]
  A graph $G=(\calV, \calE)$ is a {\em comparability graph} if edges in $\calE$ can be oriented
  so that the resulting directed graph $(\calV, \calA)$ is transitive. 
\end{definition}

Comparability graphs can be recognized in polynomial time, and the respective edge orientation 
can be computed efficiently~\cite{golumbicbook,cleanup}.

\begin{theorem}\label{thm:3comparability}
Every $3$-implementable graph is a comparability graph.
\end{theorem}
\begin{proof}
Consider a graph $G=(\calV, \calE)$ implemented by a profile $(v_1, v_2, v_3)$.
We orient the edge $\{a,b\}$ from $a$ to $b$ if $a \succ_1 b$ and from $b$ to $a$ otherwise.
Since $\{a, b\}$ is a multi-crossing pair, $a\succ_1 b$ implies $b \succ_2 a, a \succ_3 b$.
Consider a pair of arcs $(a, b)$, $(b, c)$ in the resulting directed graph.
We have $a\succ_1 b$, $b\succ_1 c$ and hence $a\succ_1 c$.
Similarly, $b\succ_2 a$, $c\succ_2 b$ implies $c\succ_2 a$ and 
           $a\succ_3 b$, $b\succ_3 c$ implies $a\succ_3 c$.
Thus, our directed graph also contains the arc $(a, c)$. 
\end{proof}

Comparability graphs are known to be perfect graphs \cite{mirsky}, 
i.e., graphs that contain neither odd-length holes nor odd-length anti-holes\footnote{%
Originally, perfect graphs are defined as graphs with the property that the chromatic number 
of every induced subgraph is equal to the size of the maximum clique in that subgraph~\cite{berge};
however, by the strong Berge conjecture, which was proved by \citet{perfect},
perfect graphs are exactly the graphs with no odd-length holes and no odd-length antiholes.}.
Hence, 
Theorem~\ref{thm:3comparability} explains why Proposition~\ref{prop:path} does not extend
to odd cycles: by definition, odd cycles are not perfect graphs.
Also, it subsumes the existence results of Propositions~\ref{prop:path}
and~\ref{prop:tree}: paths, even-length cycles, and trees can be easily seen
to be comparability graphs (we note, however, that these propositions also provide 
efficient algorithms to compute
the respective $3$-voter profiles, and it is not clear how to extract such algorithms  
from the proof of Theorem~\ref{thm:3comparability}). 
In particular, every bipartite graph is a comparability
graph (we can direct the edges from one part to the other), 
and paths, even-length cycles and trees are bipartite graphs.
However, there exists a bipartite graph that is not $3$-implementable
(and hence $3$-implementable graphs form a proper subclass of comparability graphs).  

\begin{proposition} 
  The bipartite $3$-regular graph with parts of size $4$ each (see Figure~\ref{fig:bipartite}) 
  is not $3$-implementable.
\end{proposition}
  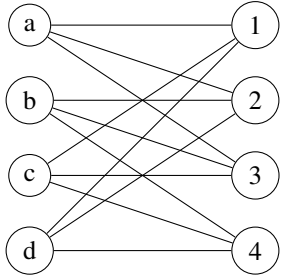
\begin{figure}
  \begin{center}
    \begin{tikzpicture}
      \begin{scope}[scale=.5,auto=left,every node/.style={circle,draw}]
        \node (n1) at (6,7)  {a};
        \node (n2) at (6,5) {b};
        \node (n3) at (6,3)  {c};
        \node (n4) at (6,1) {d};
        \node (n5) at (12,7) {1};
        \node (n6) at (12,5) {2};
        \node (n7) at (12,3) {3};
        \node (n8) at (12,1) {4};
      \end{scope}
      \begin{scope}[every node/.style={fill=white,circle}]
        \path  (n1) edge  (n5);
        \path  (n1) edge  (n6);
        \path  (n1) edge  (n7);
        \path  (n2) edge  (n6);
        \path  (n2) edge  (n7);
        \path  (n2) edge  (n8);
        \path  (n3) edge  (n7);
        \path  (n3) edge  (n8);
        \path  (n3) edge  (n5);
        \path  (n4) edge  (n8);
        \path  (n4) edge  (n5);
        \path  (n4) edge  (n6);
      \end{scope}
    \end{tikzpicture}
  \end{center}
\caption{A bipartite graph that is not $3$-implementable.\label{fig:bipartite}}
\end{figure}

\section{$\boldsymbol{n}$-Implementation for $\boldsymbol{n>3}$}\label{sec:n}
We have seen that not all graphs are $3$-implementable. 
However, we will now show that every graph $G=(\calV, \calE)$ is implementable
by an election whose number of voters is linear in $\min\{|\calV|, |\calE|\}$.
%
We first define a class of single-crossing elections that can be used to implement an arbitrary graph.

\begin{definition}\label{def:fullysc}
A single-crossing election $E = (C, V)$ with $V=(v_1, \dots, v_n)$ is
{\em fully single-crossing} if
for every pair of candidates $a, b\in C$ with $a\succ_1 b$ 
there is an $i\in[n-1]$ such that $a\succ_i b$, $b \succ_{i+1} a$, and voter $i+1$ ranks $b$
just above~$a$. 
\end{definition}

Note that in a fully single-crossing election the ranking of the last voter is the inverse
of the ranking of the first voter, i.e., every pair of candidates `crosses' exactly once.

\begin{theorem}\label{thm:full-to-imp}
If there exists a fully single-crossing election $(C, V)$ with $|C|=m$, $|V|=n$ then every $m$-vertex
graph is $(2n-1)$-implementable.
\end{theorem}
\begin{proof}
Consider a fully single-crossing election $E=(C, V)$ with $|C|=m$, $V=(v_1, \dots, v_n)$, 
and let $G=(\calV, \calE)$ be an $m$-vertex graph. Let $\widehat{V}=(v_1, v_2, v_2, \dots, v_n, v_n)$.
By construction, the election 
$\widehat{E}=(C, \widehat{V})$ is single-crossing. Now, for each edge $ab\in\calE$ we identify 
an $i\in [n-1]$ such that in $E$ we have $a\succ_i b$, $b\succ_{i+1} a$, and voter $i+1$ ranks $b$ just above $a$.
We then swap $a$ and $b$ in the preferences of the $(2i+1)$-st voter in $\widehat{V}$ 
(who, like voter $i+1$
in the original election, ranks $b$ just above $a$ prior to the swap). 
This ensures that $\{a, b\}$ is a multi-crossing pair in the 
resulting election.
In the end we obtain an election that implements $G$. 
\end{proof}

A fully single-crossing election with $m$ candidates and ${m\choose 2}+1$ voters 
can be obtained as a maximal chain in a weak Bruhat order; in this election, which we will denote by $E_B$, 
each vote differs
from its predecessor by exactly one swap of adjacent candidates (see \cite{bre-che-woe:j:single-crossing}
and the references therein). One can use $E_B$ as a starting point to implement  
an arbitrary graph $G=(\calV, \calE)$ with $2|\calE|+1$ votes: we take $E_B$, 
remove each vote that is obtained from its predecessor 
by swapping a pair of candidates that does not correspond to an edge in the input graph, 
and then use the construction in the proof of Theorem~\ref{thm:full-to-imp}. However, for dense
graphs this produces an implementation with $\Theta(|\calV|^2)$ voters. 

In contrast, our next theorem,  
in conjunction with Theorem~\ref{thm:full-to-imp}, shows that every graph 
$G=(\calV, \calE)$ is $(2|\calV|+1)$-implementable.
Our construction is inspired by the concept of odd-even sort 
\cite{ldm84} and, to the best of our knowledge, is new.

\begin{theorem}\label{thm:fullysc_m+1}
For every $m\ge 2$ there exists a fully single-crossing election with $m$ candidates and $m+1$ voters.
\end{theorem}
\begin{proof}
Let $C=\{1, \dots, m\}$.
Consider the following sequence of $m+1$ votes.
The first vote is given by $1\succ_1 2\succ_1\dots\succ_1 m$.
Then, 
for each $\ell=1, \dots, \lceil m/2 \rceil$, the vote
$2\ell$ is obtained from the vote $2\ell-1$ by swapping
the candidates in positions $2i$ and $2i-1$ 
for $i=1, \dots, \lfloor m/2 \rfloor$.
Similarly, 
for each $\ell=1, \dots, \lfloor m/2 \rfloor$, the vote
$2\ell+1$ is obtained from the vote $2\ell$ by swapping
candidates in positions $2i$ and $2i+1$ 
for $i=1, \dots, \lceil m/2 \rceil-1$.
The resulting profile for $m=7$ 
is given in Table~\ref{tbl:78}.
We will now argue that this procedure produces a fully single-crossing profile.

\begin{table}[h!]
  \begin{center}
    \begin{tabular}{|r|r|r|r|r|r|r|r|}
      \hline
      1& 2& 2& 4& 4 & 6 & 6 & 7 \\
      2& 1& 4& 2& 6 & 4 & 7 & 6 \\
      3& 4& 1& 6& 2 & 7 & 4 & 5 \\
      4& 3& 6& 1& 7 & 2 & 5 & 4 \\
      5& 6& 3& 7& 1 & 5 & 2 & 3 \\
      6& 5& 7& 3& 5 & 1 & 3 & 2 \\
      7& 7& 5& 5& 3 & 3 & 1 & 1\\
      \hline
    \end{tabular}
    \caption{A fully single-crossing election for $m=7$.}
    \label{tbl:78}
  \end{center}
\end{table}

     

First, we show that in $v_{m+1}$ the candidates are ranked as $m\succ_{m+1}\ m-1\succ_{m+1}\dots\succ_{m+1} 1$.
Indeed, suppose that $a$ is even. Then, by construction, 
for $i=2, \dots, a$, in the $i$-th vote $a$ is ranked in position $a-i+1$. 
Then, in vote $v_{a+1}$ candidate $a$ remains ranked in position $1$, 
and in the remaining $m-a$ votes $a$ moves down step by step, ending up in position $m+1-a$.
Conversely, if $a$ is odd, it moves down step by step in the first $m+1-a$ votes, stays
in the last position for one more step, and then starts climbing back up, ending in position $m+1-a$.
This implies our claim.

We have shown that each pair of candidates is swapped at least once.
To see that it is swapped exactly once, we compute the total number of swaps.
If $m$ is even, then every even-numbered vote differs from its predecessor by $\frac{m}{2}$ swaps
and every odd-numbered vote apart from the first vote differs from its predecessor by $\frac{m}{2}-1$
swaps. Thus, the total number of swaps is ${m\choose 2}$, and hence
each pair of candidates is swapped exactly once.
For odd values of $m$ the calculation is similar. 

It remains to note that, by construction, if $a\succ_i b$, but $b\succ_{i+1} a$, then $b$ is ranked
just above $a$ in $v_{i+1}$, as we only swap adjacent candidates. This completes the proof.
\end{proof}

Combining Theorem~\ref{thm:fullysc_m+1} and Theorem~\ref{thm:full-to-imp}, 
we obtain the following corollary.
\begin{corollary}\label{cor:2n+1implementation}
  An undirected graph $G=(\calV, \calE)$ is $(2|\calV|+1)$-imp\-le\-men\-table.
\end{corollary}

The bound in Corollary~\ref{cor:2n+1implementation} is linear in $|\calV|$.
One can ask if we implement each graph using a constant number of votes.
It turns out that the answer is `no'.

To show this, we use 
the Erd\"os--Szekeres theorem \cite{erdos-szekeres} 
to argue that if a graph is implementable by an
election with a few voters then it has to have a large clique or a large independent set.

\begin{lemma}\label{lem:largeclique}
If an $s$-vertex graph is $n$-implementable then it has a clique or size at least
$s^{{1}/{2^{n-1}}}$ or an independent set of size at least $s^{{1}/{2^{n-1}}}$.
\end{lemma}

On the other hand, we have the following well-known fact, 
which can be easily proved by the probabilistic method
(see, e.g., \cite{bollobas1976}).

\begin{lemma}\label{lem:smallclique}
There exists an integer constant $\alpha>0$ such that for every positive integer $s$ there exists a graph $G=(\calV, \calE)$
with $|\calV|= s$ vertices with the property that each clique and each independent set in $G$  
have at most $\alpha\log s$ vertices. 
\end{lemma}

Together, Lemmas~\ref{lem:largeclique} and~\ref{lem:smallclique} imply that for every $n\ge 0$
there are graphs that are not $n$-implementable; in fact, our proof shows that 
for each $n$ there is a graph of size at most $2^{2^{2n}}$ with this property.

\begin{theorem}\label{thm:unbounded}
For every positive integer $n$ there exists a graph $G=(\calV, \calE)$ 
with $|\calV|\le 2^{2^{2n}}$ that is not $n$-implementable.
\end{theorem}

\section{Applications}\label{sec:apps}
We will now apply the tools developed in Sections~\ref{sec:3} and~\ref{sec:n}
to the problem of detecting elections that are close to being
single-crossing with respect to a given order of voters, for two measures of closeness.

\begin{definition}\label{def:canddel}
An instance of {\sc Candidate Deletion} is given by an election $E=(C, V)$ and an integer $k \ge 1$.
It is a yes-instance if and only if there is a subset $X\subseteq C$ with $|X|\ge |C|-k$ such that
$E|_X$ is single-crossing. 
%
An instance of $k$-{\sc Candidate Partition} is given by an election $E=(C, V)$.
It is a yes-instance if and only if $C$ can be partitioned into $k$ sets $C_1, \dots, C_k$
so that for each $j\in [k]$ the election $E|_{C_j}$ is single-crossing. 
\end{definition}

We will now show that both of these problems are hard, by leveraging our observation
that $E|_X$ is single-crossing if and only if $X$ is an independent set in $\gamma(E)$.

\begin{theorem}\label{thm:canddel-n}
{\sc Candidate Deletion} is {\rm NP}-complete;
$k$-{\sc Candidate Partition}
is {\rm NP}-complete for every $k\ge 3$.
\end{theorem}
\begin{proof}
It is immediate that both of these problems are in NP. 
%
To show that {\sc Candidate Deletion} is NP-hard, we reduce from {\sc Independent Set}.
An instance of {\sc Independent Set} is given by a graph $G=(\calV, \calE)$ and an integer $t$;
it is a yes-instance if $G$ has an independent set of size at least $t$ and a no-instance otherwise.
This problem is well-known to be NP-hard~\cite{GJ}.
Given an instance $\langle G, t\rangle$ of the {\sc Independent Set} problem, we build
an election $E=(C, V)$ that implements it using the construction described in Section~\ref{sec:n};
the size of the resulting election is polynomial in the size of $G$, and $G$ has an independent
set of size at least $t$ if and only if $\langle E, t\rangle$ is a yes-instance 
of {\sc Candidate Deletion}.

We use the same argument for $k$-{\sc Candidate Partition}; the only difference is that
we reduce from the $k$-{\sc Coloring} problem. 
An instance of $k$-{\sc Coloring} is given by a graph $G=(\calV, \calE)$;
it is a yes-instance if there exists a mapping $\chi:\calV\to\{1, \dots, k\}$
such that $\chi(a)\neq\chi(b)$ for every $\{a, b\}\in\calE$ and a no-instance otherwise.
Note that each `color' $\chi^{-1}(i)$, $i\in [k]$, forms an independent set in $G$.
The $k$-{\sc Coloring} problem is well-known to be NP-hard for every $k\ge 3$~\cite{GJ}.
Again, given a graph $G$, we construct an election $E$ that implements it, 
and observe that $G$ is $k$-colorable if and only if $E$ is a yes-instance
of $k$-{\sc Candidate Partition}. 
\end{proof}

On the other hand, we can use the results in Section~\ref{sec:3} to show that
{\sc Candidate Deletion} and $k$-{\sc Candidate Partition} are in P
for elections with at most $3$ voters.

\begin{theorem}
Given an election $E=(C, V)$ with at most three voters and an integer $k$, 
we can decide in polynomial time whether the pair $\langle E, k\rangle$ is a yes-instance
of {\sc Candidate Deletion}. Also, for each $k\ge 1$
we can decide in polynomial time whether $\langle E, k\rangle$ is a yes-instance of 
$k$-{\sc Candidate Partition}.
\end{theorem}
\begin{proof}
Given an election $E=(C, V)$ with at most three voters, we construct
its multi-crossing graph $\gamma(E)$. By Theorem~\ref{thm:3comparability}
the graph $\gamma(E)$ is a comparability graph and hence a perfect graph.
As argued in the proof of Theorem~\ref{thm:canddel-n}, 
to decide whether
$\langle E, k\rangle$ is a yes-instance of {\sc Candidate Deletion}, 
it suffices to determine whether $\langle \gamma(E), |C|-k\rangle$
is a yes-instance of {\sc Independent Set}, and
to decide whether
$\langle E, k\rangle$ is a yes-instance of $k$-{\sc Candidate Partition}, 
it suffices to determine whether $\langle \gamma(E), k\rangle$
is a yes-instance of $k$-{\sc Coloring}. It remains to note that both
{\sc Independent Set} and $k$-{\sc Coloring} are known to be polynomial-time solvable
on perfect graphs (see, e.g., \citet{diestel}).
\end{proof}

By a similar argument, $2$-{\sc Candidate Partition} is polynomial-time solvable
for any number of voters.
\begin{proposition}
$2$-{\sc Candidate Partition} is polynomial-time solvable.
\end{proposition}
\begin{proof}
An election $E$ is a yes-instance of $2$-{\sc Candidate Partition}
if and only if the graph $\gamma(E)$ is $2$-colorable, 
and $2$-colorability can be checked in polynomial time.
\end{proof}

\section{Conclusions}
We have introduced the notion of single-crossing implementation of a graph 
and showed how to exploit the connection between elections and graphs 
to better understand the complexity of detecting elections that are nearly
single-crossing with respect to a fixed order of voters. Our approach turned
out to be useful for two distance measures: the number of candidates
that need to be deleted to make the input election single-crossing, and the number of parts
that the candidate set needs to be split into so that the projection
of the input election onto each set is single-crossing. There are other distance measures
that can be used in this context: e.g., we can remove or partition voters,
or swap adjacent candidates in voters' preferences. In a companion paper \cite{companion}, 
we explore the complexity of computing how far a given election is from being single-crossing 
according to several other distance measures.

Our work suggests several interesting open questions. First, 
it is not known what is the smallest value of $n$ 
such that every $m$-vertex graph is $n$-implementable: there is a significant
gap between the upper bound of Corollary~\ref{cor:2n+1implementation} and the lower
bound of Theorem~\ref{thm:unbounded}. Second, our characterization of $3$-implementable
graphs does not suggest an efficient algorithm for checking whether a graph
is $3$-implementable.
More broadly, we do not know if one can efficiently compute  
the smallest profile that implements a given graph; we conjecture that this problem is NP-complete.


\bibliographystyle{named}  
\bibliography{main}  

\begin{thebibliography}{}

\bibitem[\protect\citeauthoryear{Bachmeier \bgroup \em et al.\egroup
  }{2017}]{brandt-few}
G.~Bachmeier, F.~Brandt, C.~Geist, P.~Harrenstein, K.~Kardel, D.~Peters, and
  H.~G. Seedig.
\newblock $k$-majority digraphs and the hardness of voting with a constant
  number of voters.
\newblock Technical report, arXiv 1704.06304, 2017.

\bibitem[\protect\citeauthoryear{Berge}{1961}]{berge}
C.~Berge.
\newblock F\"arbung von {G}raphen, deren s\"amtliche bzw.~deren ungerade
  {K}reise starr sind.
\newblock {\em Wissenschaftliche Zeitschrift}, page 114, 1961.

\bibitem[\protect\citeauthoryear{Bollob{\'a}s and
  Erd{\"o}s}{1976}]{bollobas1976}
B.~Bollob{\'a}s and P.~Erd{\"o}s.
\newblock Cliques in random graphs.
\newblock {\em Mathematical Proceedings of the Cambridge Philosophical
  Society}, 80(3):419--427, 1976.

\bibitem[\protect\citeauthoryear{Bredereck \bgroup \em et al.\egroup
  }{2013}]{bre-che-woe:j:single-crossing}
R.~Bredereck, J.~Chen, and G.~Woeginger.
\newblock A characterization of the single-crossing domain.
\newblock {\em Social Choice and Welfare}, 41(4):989--998, 2013.

\bibitem[\protect\citeauthoryear{Bredereck \bgroup \em et al.\egroup
  }{2016}]{bre-che-woe:j:nice}
R.~Bredereck, J.~Chen, and G.~J. Woeginger.
\newblock Are there any nicely structured preference profiles nearby?
\newblock {\em Mathematical Social Sciences}, 79:61--73, 2016.

\bibitem[\protect\citeauthoryear{Chudnovsky \bgroup \em et al.\egroup
  }{2006}]{perfect}
M.~Chudnovsky, N.~Robertson, P.~Seymour, and R.~Thomas.
\newblock The strong perfect graph theorem.
\newblock {\em Annals of Mathematics}, 164:51--229, 2006.

\bibitem[\protect\citeauthoryear{Cornaz \bgroup \em et al.\egroup
  }{2013}]{cor-gal-spa:c:spsc-width}
D.~Cornaz, L.~Galand, and O.~Spanjaard.
\newblock Kemeny elections with bounded single-peaked or single-crossing width.
\newblock In {\em Proceedings of the 23rd International Joint Conference on
  Artificial Intelligence}, pages 76--82, 2013.

\bibitem[\protect\citeauthoryear{Diestel}{2012}]{diestel}
R.~Diestel.
\newblock {\em Graph Theory}.
\newblock Springer, 2012.

\bibitem[\protect\citeauthoryear{Elkind and Lackner}{2014}]{elk-lac:c:detect}
E.~Elkind and M.~Lackner.
\newblock On detecting nearly structured preference profiles.
\newblock In {\em Proceedings of the 28th AAAI Conference on Artificial
  Intelligence}, pages 661--667, 2014.

\bibitem[\protect\citeauthoryear{Elkind \bgroup \em et al.\egroup
  }{2017}]{ELP-trends}
E.~Elkind, M.~Lackner, and D.~Peters.
\newblock Structured preferences.
\newblock In U.~Endriss, editor, {\em Trends in Computational Social Choice},
  chapter~10, pages 187--207. AI Access, 2017.

\bibitem[\protect\citeauthoryear{Erdos and Moser}{1964}]{erdos}
P.~Erdos and L.~Moser.
\newblock On the representation of directed graphs as unions of orderings.
\newblock {\em Publications of the Mathematical Institute of the Hungarian
  Academy of Science}, 9:125--132, 1964.

\bibitem[\protect\citeauthoryear{Erd{\"o}s and Szekeres}{1935}]{erdos-szekeres}
P.~Erd{\"o}s and G.~Szekeres.
\newblock A combinatorial problem in geometry.
\newblock {\em Compositio mathematica}, 2:463--470, 1935.

\bibitem[\protect\citeauthoryear{Garey and Johnson}{1979}]{GJ}
M.~R. Garey and D.~S. Johnson.
\newblock {\em Computers and Intractability: A Guide to the Theory of
  {NP}-Completeness}.
\newblock W. H. Freeman, 1979.

\bibitem[\protect\citeauthoryear{Golumbic}{1980}]{golumbicbook}
M.~C. Golumbic.
\newblock {\em Algorithmic Graph Theory and Perfect Graphs}.
\newblock Academic Press [Harcourt Brace Jovanovich, Publishers], New
  York-London-Toronto, Ont., 1980.
\newblock With a foreword by Claude Berge, Computer Science and Applied
  Mathematics.

\bibitem[\protect\citeauthoryear{Jaeckle \bgroup \em et al.\egroup
  }{2018}]{jae-pet-elk:c:nearly-sc}
F.~Jaeckle, D.~Peters, and E.~Elkind.
\newblock On recognising nearly single-crossing preferences.
\newblock In {\em Proceedings of the 32nd AAAI Conference on Artificial
  Intelligence}, pages 1079--1086, February 2018.

\bibitem[\protect\citeauthoryear{Lakhani \bgroup \em et al.\egroup
  }{2019}]{companion}
F.~Lakhani, D.~Peters, and E.Elkind.
\newblock Correlating preferences and attributes: nearly single-crossing
  profiles.
\newblock In {\em Proceedings of the 28th International Joint Conference on
  Artificial Intelligence}, 2019.

\bibitem[\protect\citeauthoryear{Lakshmivarahan \bgroup \em et al.\egroup
  }{1984}]{ldm84}
S.~Lakshmivarahan, S.~K. Dhall, and L.~L. Miller.
\newblock Parallel sorting algorithms.
\newblock In {\em Advances in Computers}, volume~23, pages 295--354. Elsevier,
  1984.

\bibitem[\protect\citeauthoryear{McGarvey}{1953}]{mcgarvey}
D.~C. McGarvey.
\newblock A theorem on the construction of voting paradoxes.
\newblock {\em Econometrica}, 21(4):608--610, 1953.

\bibitem[\protect\citeauthoryear{Mirrlees}{1971}]{mir:j:single-crossing}
J.~Mirrlees.
\newblock An exploration in the theory of optimal income taxation.
\newblock {\em Review of Economic Studies}, 38:175--208, 1971.

\bibitem[\protect\citeauthoryear{Mirsky}{1971}]{mirsky}
L.~Mirsky.
\newblock A dual of {D}ilworth's decomposition theorem.
\newblock {\em The American Mathematical Monthly}, 78(8):876--877, 1971.

\bibitem[\protect\citeauthoryear{Roberts}{1977}]{rob:j:tax}
K.~W.~S. Roberts.
\newblock Voting over income tax schedules.
\newblock {\em Journal of Public Economics}, 8(3):329--340, 1977.

\bibitem[\protect\citeauthoryear{Simon and Trunz}{1994}]{cleanup}
K.~Simon and P.~Trunz.
\newblock A cleanup on transitive orientation.
\newblock In {\em Orders, algorithms, and applications ({L}yon, 1994)}, volume
  831 of {\em Lecture Notes in Comput. Sci.}, pages 59--85. Springer, Berlin,
  1994.

\bibitem[\protect\citeauthoryear{Skowron \bgroup \em et al.\egroup
  }{2015}]{sko-yu-fal:j:mwsc}
P.~Skowron, L.~Yu, P.~Faliszewski, and E.~Elkind.
\newblock The complexity of fully proportional representation for
  single-crossing electorates.
\newblock {\em Theoretical Computer Science}, 569:43--57, 2015.

\bibitem[\protect\citeauthoryear{Stearns}{1959}]{stearns}
R.~Stearns.
\newblock The voting problem.
\newblock {\em The American Mathematical Monthly}, 66(9):761--763, 1959.

\end{thebibliography}

\end{document}